\documentclass[12pt, final]{l4dc2023}

\usepackage{etoolbox,environ}
\newtoggle{short}
\togglefalse{short}  %
\NewEnviron{extendedonly}
  {\iftoggle{short}{}{\BODY}}
\NewEnviron{shortonly}
  {\iftoggle{short}{\BODY}{}}
  
\newcommand{\R}{\mathbb{R}}
\newcommand{\C}{\mathbb{C}}
\newcommand{\Z}{\mathbb{Z}}
\newcommand{\N}{\mathbb{N}}
\newcommand{\simiid}{\overset{\textrm{i.i.d.}}{\sim}}

\newcommand{\Norm}[1]{\mathcal{N}\left(#1\right)}
\newcommand{\CNorm}[1]{\mathcal{CN}\left(#1\right)}
\newcommand{\Ex}[1]{\mathbb{E}\left[#1\right]}
\DeclareMathOperator{\trace}{trace}
\DeclareMathOperator{\ReOp}{Re}
\renewcommand{\Re}[1]{\ReOp\left[#1\right]}
\newcommand{\tRe}[1]{\ReOp[#1]}

\newcommand{\aug}[1]{\underline{#1}}

\newenvironment{proofsketch}{\paragraph{\textbf{Proof sketch}}}{\hfill$\blacksquare$}

\usepackage{mathtools}
\DeclarePairedDelimiterX{\norm}[1]{\lVert}{\rVert}{#1}
\DeclarePairedDelimiterX{\normi}[1]{\lVert}{\rVert_{\infty}}{#1}

\title{Frequency Domain Gaussian Process Models for $H^\infty$ Uncertainties}
\usepackage{times}

\author{%
 \Name{Alex Devonport} \Email{alex\_devonport@berkeley.edu}\\
 \addr Department of Electrical Engineering and Computer Sciences \\
       University of California, Berkeley
 \AND
 \Name{Peter Seiler} \Email{pseiler@umich.edu}\\
 \addr Department of Electrical Engineering and Computer Science \\
       University of Michigan, Ann Arbor
 \AND
 \Name{Murat Arcak} \Email{arcak@berkeley.edu}\\
 \addr Department of Electrical Engineering and Computer Sciences \\
       University of California, Berkeley
}

\begin{document}

\maketitle

\begin{abstract}%
    Complex-valued Gaussian processes are used in Bayesian frequency-domain
    system identification as prior models for regression. If each realization
    of such a process were an $H_\infty$ function with probability one, then the
    same model could be used for probabilistic robust control, allowing for
    robustly safe learning.
    We investigate sufficient conditions for a general complex-domain Gaussian
    process to have this property. 
    For the special case of processes whose
    Hermitian covariance is stationary, we provide an explicit parameterization
    of the covariance structure in terms of a summable sequence of nonnegative
    numbers.
\end{abstract}


\section{Introduction}

With the general popularity of Gaussian process models in machine
learning and in particular their growing adoption in data-driven
control, there have been recent advances in using Gaussian process models as
nonparametric Bayesian estimators in system identification. Initially this was
done in the time domain, with works 
like~\cite{pillonetto2010new} 
and~\cite{chen2012estimation}
using Gaussian processes to identify the impulse
response of an LTI stable system. Subsequent works consider frequency-domain
regression,
such as~\cite{lataire2016transfer} 
which uses a modified complex Gaussian process regression
model to estimate transfer functions from discrete Fourier transform (DFT) data, 
and~\cite{stoddard2019gaussian}
which considers a
similar regression approach to estimate the generalized frequency response of
nonlinear systems.
These methods also have close ties to some non-probabilistic estimation methods,
such as analytic interpolation (\cite{singh2020loewner,takyar2010analytic}) and kernel-based interpolation
(\cite{khosravi2021kernel,khosravi2019kernel}). At the heart of these Bayesian
techniques is the \emph{prior model}, a probabilistic dynamical model of an uncertain
system that represents one's knowledge of the system
prior to collecting any data.

Probabilistic dynamical models for uncertain systems are also used extensively in probabilistic robust
control, such as probabilistic $\mu$
analysis~(\cite{khatri1998guaranteed,balas2012analysis,biannic2021advanced}),
disk margins~(\cite{somers2022probabilistic}), and the methods reviewed
in~\cite{calafiore2007probabilistic}.
In probabilistic robust control, each
possible realization of the probabilistic uncertainty must be interpretable as a
system of the type being modeled; otherwise, robustness guarantees involving
ensembles of uncertainties would not be meaningful. This is a strong
interpretability requirement compared to Bayesian system identification, where
typically only the regression mean needs to be interpretable.

Since both Bayesian system identification and probabilistic robust control use
probabilistic uncertainty models, applying both techniques to the same model is
a promising strategy for safely learning an unknown or uncertain control system.
However, this is only possible if the nonparametric uncertainty model used for
system identification satisfies the stronger interpretability requirement of
probabilistic robust control. Prior work in Bayesian system identification 
like~\cite{pillonetto2010new} and~\cite{lataire2016transfer} can guarantee that
the regression mean is causal and stable, but guarantees on the mean do not
generally carry over to guarantees about process realizations.
For example, while the predictive mean of a Gaussian
process regression model is guaranteed to inhabit the reproducing kernel Hilbert space (RKHS)
corresponding to the prior covariance, it is generally the case that the
realizations of the Gaussian process model itself lie outside of the RKHS of its
covariance with probability one.

The contribution of this paper is to provide conditions under which a Gaussian
process model is appropriate for both Bayesian system identification and
probabilistic control. Specifically, 
we provide conditions under which realizations of a complex Gaussian process of
a complex variable correspond to the z-transform of an LTI, causal, BIBO
stable, and real system with probability one. Since an LTI, causal, and BIBO stable system is
characterized by a z-transform that resides in the Hardy space $H_\infty$, we
refer to such processes as $H_\infty$ Gaussian processes.
We give the conditions
(Theorem~\ref{prop:hinf-conditions} and Proposition~\ref{prop:real-impulse}) 
directly in terms of the frequency-domain covariance functions: this
allows one to design frequency-domain covariance functions directly, as opposed
to the approach used by prior works in Bayesian system identification, where
frequency-domain covariances must be derived from the z-transform (or Laplace
transform) of a time-domain stochastic impulse response. In cases where prior
knowledge is given in frequency-domain terms, being able to construct the
frequency-domain covariance is more practical.

In addition to the general conditions, we provide a complete characterization
(Theorem~\ref{prop:stationary-process}) of
the covariance structure of a special class of $H_\infty$ Gaussian process,
namely those whose Hermitian covariance is stationary. Each Hermitian stationary
$H_\infty$ process is parameterized by a summable sequence of
nonnegative reals, which lead to computationally tractable closed forms for
certain choices of sequences. Since stationary processes are a popular choice for
GP regression priors, this characterization makes it possible to construct useful and
computationally convenient priors for Bayesian system identification that are
also fully interpretable as probabilistic dynamical models.

To verify the utility of $H_\infty$ GP models for Bayesian transfer function
estimation, we apply the technique to two second-order systems using a mixture
of a Hermitian stationary $H_\infty$ processes constructed with
Theorem~\ref{prop:stationary-process} and an $H_\infty$ process designed to
model resonance peaks. Contrary to other recent work in Bayesian system
identification, we choose to use the \emph{strictly linear estimator} for our
Gaussian process models instead of the \emph{widely linear estimator}.
Although the widely linear estimate is superior for general processes, we
find that for $H_\infty$ Gaussian process models the strictly linear estimator
works nearly as well while being simpler and more stable to compute than the
widely linear estimator.
\begin{extendedonly}
\end{extendedonly}

The rest of the paper is organized as follows.
Section~\ref{sec:preliminaries} introduces the system setup, reviews
background information on complex-valued random variables and stochastic
processes, and introduces the classes of complex Gaussian processes that we
study in this paper. Section~\ref{sec:hinf-processes} provides the
conditions and characterizations described in the last paragraph,
and represents the main technical contribution of this work.
Section~\ref{sec:regression} 
reviews widely linear and strictly linear complex estimators for complex
Gaussian process regression and presents numerical examples of Bayesian system
identification.
\begin{shortonly}
We omit the full proofs of
of Theorem~\ref{prop:hinf-conditions},
Proposition~\ref{prop:real-impulse}, and Theorem~\ref{prop:stationary-process} in this paper
in favor of ``proof sketches'' for
brevity. The full proofs are available in an extended paper~(EXTENDED PAPER)
available online.
\end{shortonly}

\section{Notation}

For a complex vector or matrix $X$, $X^*$ denotes the complex conjugate and
$X^H$ denotes the conjugate transpose.
We denote the exterior of the unit disk as 
$E = \{Z\in\C : |z| > 1\}$, and its closure as
$\bar{E} = \{Z\in\C : |z| \ge 1\}$.
$L_2$ is the Hilbert space of functions $f:\C\to\C$ such that
$\int_{-\pi}^{\pi} |f(Re^{j\Omega})|^2d\Omega < \infty$, 
equipped with the
inner product $\langle f,g\rangle_2 = \int_{-\pi}^{\pi}
f(e^{j\Omega}){g^*}(e^{j\Omega})d\Omega$.
$H_2$ is the Hilbert space of functions $f:\C\to\C$ that are bounded
and analytic for all $z\in E$ such satisfy and 
$\int_{-\pi}^{\pi} |f(Re^{j\Omega})|^2d\Omega < \infty$,
equipped with the
inner product $\langle f,g\rangle_2 = \int_{-\pi}^{\pi}
f(e^{j\Omega}){g^*}(e^{j\Omega})d\Omega$. It is a vector subspace of
$L_2$.
$H_\infty$ is the Banach space of functions $f:\bar{E}\to\C$ that are bounded
and analytic for all $z\in E$ and
$\sup_{\Omega\in[-\pi,\pi]} |f(e^{j\Omega})| < \infty$, equipped with the
norm $\normi{f} = \sup_{\Omega\in[-\pi,\pi]} |f(e^{j\Omega})|$.
$\ell^1$ is the space of absolutely summable sequences, 
that is sequences $\{a_n\}_{n=0}^\infty$ 
such that $\sum_{n=0}^\infty |a_n| <\infty$.

\section{Preliminaries}
\label{sec:preliminaries}

The object of this paper is to construct nonparametric statistical models for
causal, LTI, BIBO stable systems in the frequency domain. Since our main focus
will be the probabilistic aspects of the model, we restrict our attention to the
simplest dynamical case: a single-input single-output system in discrete time.
Thus, our dynamical systems are frequency-domain multiplier operators 
$H_f:L_2\to L_2$ whose output is defined pointwise as $(H_f
u)(\omega)=f(\omega)u(\omega)$, 
where $f:\C\to\C$ is the system's transfer function. Thanks to the bijection 
$H_f \leftrightarrow f$, we generally mean the function $f$ when we refer
to ``the system''.

Since our aim is to construct a probabilistic model for the system that is not
restricted to a finite number of parameters, we must work directly with random
complex functions of a complex variable: this is a special type of complex
stochastic process which we call a z-domain process.

\begin{definition}
    \label{def:z_domain_process}
    Let $(\Xi, F, \mathbb{P})$ denote a probability space. 
    A z-domain stochastic process with domain $D\subseteq \C$ is a
    function $f:\Xi\times D\to\C$.
\end{definition}

Note that each value of $\xi\in\Xi$ yields a function
$f_\xi=f(\xi,\cdot):D\to\C$, which is called either a ``realization'' or a
``sample path'' of $f$.
If we take $\xi$ to be selected at random according to the probability law
$\mathbb{P}$, then $f_\xi$ represents a ``random function'' in the frequentist
sense.
Alternatively, if we have a prior belief
about the likelihood of some $f_\xi$ over others, 
we may encode this belief in a Bayesian sense using the measure $\mathbb{P}$.
We drop the dependence of $f$ on $\xi$ from the notation outside of definitions,
as it will be clear when $f(z)$ refers to the random variable $f(\cdot,z)$ or
when $f$ stands for a realization $f_\xi$.

\begin{shortonly}
\begin{definition}
    \label{def:gaussian_process}
    A Gaussian z-domain process is a z-domain process $f$ such that, for any
    $n>0$, the random vector
    $(f(z_1),\dotsc,f(z_n))$ is complex multivariate Gaussian 
    for all $(z_1,\dotsc,z_n)\in D^n$.
\end{definition}
A complex Gaussian process is more than two real-valued Gaussian processes added
together, as the real and imaginary parts may depend on each other. Unlike a
real Gaussian process, which is characterized by its mean
$m(t)=\Ex{x(t)}$ and covariance $k(t,s)=\Ex{x(t)x(s)}$, a complex Gaussian
process $f$ is characterized by three functions: its mean $m(z)=\Ex{f(z)}$,
its \emph{Hermitian covariance} $k(z,w)=\Ex{f(z)f^*(w)}$ and its
\emph{complementary covariance} $\tilde{k}(z,w)=\Ex{f(z)f(w)}$.
\end{shortonly}

\begin{extendedonly}

\subsection{Complex Random Variables and Stochastic Processes}%
\label{sub:complex_random_variables_and_stochastic_processes}

Complex random variables and processes are essentially no different from real
ones, but certain statistical descriptions for real random variables do not
extend to the complex case unless suitably augmented.
The following is an example of what can go wrong.
\begin{example}
    A real Gaussian random variable is completely
    determined by its mean and variance. However, this is not true for complex
    Gaussian random variables. Consider the random variables
    $Z = X+jY$, and $W=j\sqrt{2}X$, where $X,Y\simiid\Norm{0,1}$. These are
    distinct random variables, as evidenced by the fact that they have different
    supports; however, their means are $\Ex{Z}=\Ex{W}=0$, and their variances
    are $\Ex{ZZ^*}=\Ex{WW^*}=2$.
\end{example}
We have defined the variance of a mean-zero complex random variable
$Z$ to be $\Ex{ZZ^*}$: this is required in order for the variance to
be a real nonnegative number, and specializes to the standard variance in the
purely real case.

What statistic, not required in the real case, distinguishes $Z$ and $W$? It
turns out to be the ``real'' variance $\Ex{Z^2}$; in the example, we have
$\Ex{Z^2}=0$ and $\Ex{W^2}=-2$. In general, a complex Gaussian $Z$ is completely
specified by $\Ex{Z}$, $\Ex{ZZ^*}$, and $\Ex{Z^2}$. We call 
the statistic 
$\sigma^2=\Ex{ZZ^*}$ the \emph{Hermitian variance}, and
$\tilde{\sigma}^2=\Ex{Z^2}$ the \emph{complementary variance}.%
\footnote{Other names for the complementary (co)variance in the literature of
complex random variables and stochastic processes are the
\emph{pseudo(co)variance} and \emph{relation}.}
We carry this
nomenclature to stochastic processes, assigning to a z-domain process the
Hermitian covariance function $k(z,w)=\Ex{f(z)f^*(w)}$ and complementary
covariance function $\tilde{k}(z,w)=\Ex{f(z)f(w)}$.

\begin{definition}
    \label{def:gaussian_process}
    A Gaussian z-domain process is a z-domain process $f$ such that, for any
    $n$, the random vector
    $(f(z_1),\dotsc,f(z_n))$ is complex multivariate Gaussian-distributed
    for all $(z_1,\dotsc,z_n)\in D^n$.
\end{definition}
Analogous to the way that a real Gaussian process is determined by its mean and
covariance, a Gaussian z-domain process is completely specified by its mean
$m:D\to\C$,
Hermitian covariance $k:D\times D\to\C$, 
and complementary covariance $\tilde{k}:D\times D\to \C$. 

A complex random variable $f=x+jy$ may also be represented in its \emph{augmented form}
$\aug{f}=(f,f^*)$. This form is useful despite being redundant, as it yields a convenient expression of the
second-order statistics in terms of an \emph{augmented covariance matrix}
\begin{equation}
    \aug{K}_f 
    = \Ex{
    \begin{bmatrix}
        f \\ f^*
    \end{bmatrix}
    \begin{bmatrix}
        f^* & f
    \end{bmatrix}
    }
    =
    \begin{bmatrix}
        K_f & \tilde{K}_f \\ \tilde{K}_f^* & K_f^*
    \end{bmatrix}
\end{equation}
where $K_f$, $\tilde{K}_f$ are the Hermitian and complementary covariances of
$f$. We use the augmented form in some proofs in the following section.
Similarly, the second-order statistics of a complex process can be
expressed by the
\emph{augmented covariance function}
\begin{equation}
    \aug{k}(z,w) =
    \Ex{
    \begin{bmatrix}
        f(z) \\ f^*(z)
    \end{bmatrix}
    \begin{bmatrix}
        f^*(w) & f(w)
    \end{bmatrix}
    }
    =
    \begin{bmatrix}
        k(z,w) & \tilde{k}(z,w) \\ \tilde{k}^*(z,w) & k^*(z,w)
    \end{bmatrix}
\end{equation}
Generally, an \aug{underline} denotes an augmented representation,
either of a process or of a covariance function or matrix.
\end{extendedonly}

\subsection{$H_\infty$ Gaussian Processes}%
\label{sub:z_domain_and_h_infty_gaussian_processes}

Consider a deterministic input-output operator $H_g$ with transfer function
function $g:D\to\C$. The condition that $H_g$ belong to the operator space
$H^\infty$ of LTI, causal, and BIBO stable systems is that $g$ belong to the
function space $H_\infty$. Now suppose we wish to construct 
a random operator $H_f$ using the realizations of a z-domain process $f$ as its
transfer function: the analogous condition is that the realizations of $f$
lie in $H_\infty$ with probability one.

\begin{definition}
    \label{def:hinf_process}
    A z-domain process is called an $H_\infty$ process when the set
    $\{\xi\in\Xi : f_\xi \in H_\infty\}$ has measure one under $\mathbb{P}$.
\end{definition}

Less formally, an $H_\infty$ process is a z-domain process $f$ such that
$\mathbb{P}(f\in H_\infty)=1$. Having $f_\xi\in H_\infty$ implies that
$\bar{E}\subseteq D$: we usually take $D=\bar{E}$.
If we also require that $H_g$ give real
outputs to real inputs in the time domain, $g$ must satisfy the
conjugate symmetry relation $g(z^*)=g^*(z)$ for all $z\in D$. The analogous
condition for $H_f$ is to require that $f$ satisfy the condition with
probability one.

\begin{definition}
    \label{def:conjugate_symmetric}
    A z-domain process $f$ is called conjugate symmetric when the set
    $\{\xi\in\Xi : f_\xi(z^*)=f_\xi^*(z),\ \forall z\in D\}$ has measure one under $\mathbb{P}$.
\end{definition}

Combining definitions~\ref{def:gaussian_process},~\ref{def:hinf_process},
and~\ref{def:conjugate_symmetric}, we
arrive at our main object of study: conjugate-symmetric $H_\infty$ Gaussian
processes. 

\begin{example}[``Cozine'' process]
    \label{ex:cozine}
    The random transfer function
    \begin{equation}
        \label{eq:cozine_form}
        f(z) = \frac{X - a(X\cos(\omega_0) - Y\sin(\omega_0))z^{-1}}
                    {1-2a\cos(\omega_0)z^{-1} + a^2z^{-2}},
    \end{equation}
    where $X,Y\simiid\Norm{0,1}$, $a\in(0,1)$, $\omega_0\in[0,\pi]$,
    is a z-domain Gaussian process. From the form of the transfer function, we
    see that $H$ is bounded on the unit circle, analytic on $E$, and conjugate
    symmetric with probability one, from which it follows that $f$ is 
    a conjugate symmetric $H_\infty$ process.
    Since $f$ corresponds to the z-transform of an exponentially decaying
    discrete cosine with random magnitude and phase, we call it a
    \emph{``cozine'' process}.
    Its Hermitian and complementary covariances are
    \begin{equation}
        \begin{aligned}
        \label{eq:cozine_covariances}
            k(z,w) 
            &= 
            \frac{1-a\cos(\omega_0)(z^{-1}+(w^*)^{-1}) + a^2(zw^*)^{-1}}
            {(1-2a\cos(\omega_0)z^{-1} + a^2z^{-2})(1-2a\cos(\omega_0)(w^*)^{-1} + a^2(w^*)^{-2})},
            \\
            \tilde{k}(z,w) 
            &= 
            \frac{1-a\cos(\omega_0)(z^{-1}+w^{-1}) + a^2(zw)^{-1}}
            {(1-2a\cos(\omega_0)z^{-1} + a^2z^{-2})(1-2a\cos(\omega_0)w^{-1} + a^2w^{-2})}.
        \end{aligned}
    \end{equation}
\end{example}
As a Bayesian prior for an $H^\infty$ system, this process represents a belief
that the transfer function exhibits a
resonance peak (of unknown magnitude) at $\omega_0$. Knowing $\omega_0$ in
advance is a strong belief, but it can be relaxed by taking a hierarchical model
where $\omega_0$ enters as a hyperparameter. When used as a prior, the
hierarchical model
represents the less determinate belief that there is a resonance peak
\emph{somewhere}, whose magnitude can be made arbitrarily small if no peak is
evident in the data. 

The construction in Example~\ref{ex:cozine}, where properties of conjugate
symmetry and BIBO stability can be checked directly, may be extended to random
transfer functions of any finite order. However, the technique does not carry to
the infinite-order $H_\infty$ processes required for nonparametric Bayesian
system identification, or more generally for applications that do not place an \emph{a
priori} restriction on the order of the system. 
We are therefore motivated to find conditions under which a z-domain process is
a conjugate-symmetric $H_\infty$ Gaussian process expressed directly in
terms of $k$ and $\tilde{k}$.

\section{Constructing $H_\infty$ Gaussian Processes}
\label{sec:hinf-processes}

\begin{extendedonly}
In this section, we consider a z-domain Gaussian process $f$ with zero mean,
Hermitian covariance function $k$, and complementary covariance function
$\tilde{k}$. Taking zero mean implies no loss in generality: to lift any of
these conditions to a process with nonzero mean, we simply ask that the desired
property (inhabiting $H_\infty$, possessing conjugate symmetry, or both) also
hold for the mean.
\end{extendedonly}

\begin{extendedonly}

Our first step towards finding conditions under which a z-domain process is a
conjugate-symmetric $H_\infty$ process is the
observation that
all functions in $H_\infty$ are also in $H_2$: $\normi{f}<\infty$ implies that 
$\int_{-\pi}^{\pi} |f(Re^{j\Omega})|^2d\Omega$
converges for all $R\ge 1$. 
Indeed, $f\in H_\infty$ precisely when $f\in H_2$ and $\normi{f}<\infty$. Our
general strategy for proving that a z-domain process is an $H_\infty$ process
$f$ is to show that $f_\xi\in H_2$ and $\normi{f_\xi}<\infty$ hold for realizations
$f_\xi$ of the process with probability one.

To show when $f\in H_2$ with probability one,
we use the fact that $H_2$ is a reproducing kernel Hilbert space.
\begin{definition}
    A reproducing kernel Hilbert space (RKHS) is a Hilbert space of functions
    on a domain $D$ for which the \emph{evaluation functionals} $E_z$, defined
    pointwise as $E_z f = f(z)$, are bounded in the sense that 
    $E_z f \le M(z) \norm{f}$ for some nonnegative function $M(z)$.
\end{definition}
Applying the Riesz representation theorem to the evaluation functionals, which
are linear and by assumption bounded, we recover the \emph{reproducing kernel}
$k$ that satisfies $\langle k(z,\cdot),f\rangle = f(z)$ and that $k(z,z)$ is the least $M(z)$
such that $E_z f \le M(z) \norm{f}$ holds.
The form of $k$ can be derived from an orthonormal basis for the RKHS using the
following result.

\begin{lemma}[\cite{RKHS}, Theorem 2.4]
    \label{lem:basis_kernel}
    Let $H$ be an RKHS with reproducing kernel $k$. If $e_0,e_1,\dotsc$ form an
    orthonormal basis for $H$, then
    $k(z,w) = \sum_{n=0}^{\infty} e_n(z)e^*_n(w)$ where the series converges
    pointwise.
\end{lemma}

Discrete-time $H_2$ is a type of \emph{Hardy class}, which is a space of complex
functions that are analytic on a domain of the complex plane and satisfy a
bounded-growth condition on the boundary. When this domain is a
half-plane or the interior of the unit disk, it is well known that these spaces
are RKHSs. It is therefore not surprising that the same is true when the domain
is the exterior of the unit disk, as it is for our $H_2$. However, we are not
aware of a citable proof of this fact, nor of formula for its kernel, so we
provide both here for completeness.

\begin{proposition}
    \label{thm:h2_rkhs}
    Discrete-time $H_2$ is a reproducing kernel Hilbert space with kernel
    function
    \begin{equation*}
        k(z,w)=\frac{zw^*}{zw^*-1}
    \end{equation*}
    and orthonormal basis $\{e_n\}_{n=0}^\infty$,  $e_n(z) = z^{-n}$.
\end{proposition}

\begin{proof}
    The first step is to compute the orthonormal basis. We begin with the standard fact 
    (see~\cite[theorem 13.3]{young1988hilbert})
    that $e_n(z) = z^{-n}, n\in\Z$ form an orthonormal basis for $L_2$. 
    Since $H_2$ is a
    subspace of $L_2$, we can take the orthogonal projection of $z^{-n}, n\in\Z$
    onto $H_2$ to yield a sequence that spans $H_2$. For $n<0$, $z^{-n}$ is
    unbounded on the exterior of the unit disk, so $P(e_n) = 0$ for $n<0$. On the other
    hand, $z^{-n}\in H_2$ for $n\ge0$, so $P(e_n) = e_n$ for $n\ge 0$.
    Discarding the zeros, we have that $H_2=\text{span}(\{e_n\}_{n=0}^\infty)$.
    Since $L_2$ and $H_2$ have the same norm, we already know that $e_n$, $n\ge
    0$ are orthonormal in $H_2$. The combined facts of orthonormality and
    spanning the space ensure 
    (e.g. by~\cite[chapter 2, \S 8, Theorem 3]{helmberg1969hilbert}) 
    that $e_n, n\ge 0$ form an orthonormal basis for $H_2$.

    Now we establish that the evaluation functionals $E_z(f) = f(z)$ are
    bounded. Let $f\in H_2$: from the paragraph above, we expand $f$ as
    $f(z) = \sum_{n=0}^\infty a_n z^{-n}$. By the Parseval identity
    $\norm{f}_2=\sum_{n=0}^\infty |a_n|^2$, we have
    \begin{equation}
        \begin{aligned}
            \label{eq:}
            |f(z)| = \left|\sum_{n=0}^\infty a_n z^{-n}\right|
            \le \sum_{n=0}^\infty |a_n| |z|^{-n}
            \le \sqrt{\sum_{n=0}^\infty |a_n|^2}
            \sqrt{\sum_{n=0}^\infty (|z|^{-2})^{n}}
            = \norm{f}_2 \frac{1}{\sqrt{1-|z|^{-2}}}.
        \end{aligned}
    \end{equation}
    Since $E_z(f) \le M(z)\norm{f}_2$ for $M(z)=1/\sqrt{1-|z|^{-2}}$, it follows that $H_2$ is an RKHS.
    This allows us to apply Lemma~\ref{lem:basis_kernel} to compute
    \begin{equation}
          k(z,w)
          = \sum_{n=0}^\infty e_n(z)e^*_n(w)
          = \sum_{n=0}^\infty z^{-n}(w^*)^{-n}
          = \frac{1}{1-(zw^*)^{-1}}
          = \frac{zw^*}{zw^* - 1}.
    \end{equation}

\end{proof}

The fact that $H_2$ is an RKHS allows us to use Driscoll's zero-one theorem to
establish if the realizations of a z-domain Gaussian process belong to $H_2$.

\begin{lemma}[\cite{driscoll1973reproducing}]
    \label{lem:driscoll}
    Let $f$ be a mean zero Gaussian process on a parameter set $T$ with covariance
    function $k$. Let $r$ be the reproducing kernel of an RKHS of functions with
    domain $T$. Let $t_1,t_2,\dotsc$ denote a countably dense set of points in
    $T$, and define 
    $K^n, R^n \in \R^{n\times n}$
    as $(K_n)_{ij}=k(t_i, t_j)$, $(R_n)_{ij}=r(t_i,t_j)$. Then the realizations
    of $f$ are in the RKHS with kernel $k$ with probability either zero or one,
    according respectively to whether $\sum_n \trace K_nR_n^{-1}$ is infinite or
    finite.
\end{lemma}

To ensure that $\normi{f}<\infty$, we need a sufficient condition under which
the realizations of $f$ are bounded on the unit circle.
The following result provides a sufficient condition in
terms of the continuity of the covariance.

\begin{lemma}[\cite{RFG}, Theorem 1.4.1]
    \label{lem:boundedness_condition}
    Let $f$ be a real-valued Gaussian process with mean zero defined on a
    compact parameter set $T\subseteq R^n$. If there exist positive constants
    $C$, $\alpha$, and $\delta$ such that the covariance function $k$ satisfies
    \begin{equation}
        k(s,t) = k(s,s) + k(t,t) - 2k(s,t)
        \le \frac{C}{|\log|\theta-\phi||^{1+\alpha}}\\
    \end{equation}
    for $s,t\in T$ such that $|s-t|<\delta$, then 
    \begin{equation}
        P\left(\sup_{t\in T} |f(t)| < \infty \right) = 1.
    \end{equation}
\end{lemma}
We are now prepared to return to $H_\infty$ Gaussian processes.\ 
\end{extendedonly}
 The following result provides the general test to determine if $f$ is an
$H_\infty$ Gaussian process, by establishing with probability one that $f_\xi\in
H_2$ and $\normi{f_\xi}<\infty$. 

\begin{theorem}
    \label{prop:hinf-conditions}
     Let $f$ be a z-domain Gaussian process with mean zero and
     continuous Hermitian covariance $k$ and complementary covariance $\tilde{k}$.
     Let 
     $k_r=\tfrac{1}{2}\tRe{k+\tilde{k}}$, 
     $k_i=\tfrac{1}{2}\tRe{k-\tilde{k}}$
     denote the covariance functions of the real and imaginary parts of $f$
     respectively.
     Then $f$ is an $H_\infty$ process
     under the
     following conditions:
     \begin{enumerate}
         \item There exist positive, finite constants
             $C_r$, $C_i$, $\alpha_r$, $\alpha_i$, $\delta_r$, $\delta_i$,
             such that
             $k_r$ and $k_i$, restricted to the unit circle, satisfy the
             following continuity conditions:
             \begin{equation}
                 \begin{aligned}
                     \label{eq:hinf-conditions-boundedvar}
                       k_r(e^{j\theta},e^{j\theta})
                     + k_r(e^{j\phi},e^{j\phi})
                     - 2 k_r(e^{j\theta},e^{j\phi})
                     &\le \frac{C_r}{|\log|\theta-\phi||^{1+\alpha_r}}
                     \quad
                     \forall |\theta-\phi| < \delta_r
                     \\
                       k_i(e^{j\theta},e^{j\theta})
                     + k_i(e^{j\phi},e^{j\phi})
                     - 2 k_i(e^{j\theta},e^{j\phi})
                     &\le \frac{C_i}{|\log|\theta-\phi||^{1+\alpha_i}}
                     \quad 
                     \forall |\theta-\phi| < \delta_i.
                 \end{aligned}
             \end{equation}
         \item 
             Let $\{z_n\}_{n=1}^\infty$ be a countable dense sequence of points in
             $E$. For $n\in\N$, define the Gramian
             matrices
             $K_r^n, K_i^n, R^n \in \R^{n\times n}$
             as 
             $(K_r^n)_{jl}=k_r(z_j, z_l)$, 
             $(K_i^n)_{jl}=k_i(z_j, z_l)$, 
             and
             $(R^n)_{jl}=r(z_j,z_l)$, where $r(z_j,z_l)=z_iz^*_j/(z_iz^*_j - 1)$.
             $K_r^n$, $K_i^n$, and $R^n$ satisfy
     \begin{equation}
         \begin{aligned}
             \label{eq:hinf-conditions-driscoll}
             \sup_{n\in\N}\trace K_r^n (R^n)^{-1} < \infty
             \quad\mbox{and}\quad
             \sup_{n\in\N}\trace K_i^n (R^n)^{-1} < \infty.
         \end{aligned}
     \end{equation}
     \end{enumerate}
\end{theorem}

\begin{shortonly}
\begin{proofsketch}
    Discrete-time $H_2$ is an RKHS with kernel
    $r(z,w)=zw^*/(zw^*-1)$. (This fact is proven in~(EXTENDED PAPER).)
    Condition~\eqref{eq:hinf-conditions-driscoll} then ensures by Driscoll's
    zero-one theorem~(\cite{driscoll1973reproducing}) that the sample paths of
    $f$ inhabit $H_2$ with probability one. 
    Condition~\eqref{eq:hinf-conditions-boundedvar} ensures 
    by~\cite[Theorem 1.4.1]{RFG} that the restriction of $f$ to the unit circle
    is bounded with probability one. Since an $H_\infty$ function is precisely
    an $H_2$ function whose values on the unit circle are bounded, it follows
    that $f$ inhabits $H_\infty$ with probability one.
\end{proofsketch}
\end{shortonly}

\begin{extendedonly}

\begin{proof}
    To apply Lemmas~\ref{lem:driscoll} and~\ref{lem:boundedness_condition}, we work separately with the real and
    imaginary parts of the process. To that end, we write
    $f = x + jy$, where $x$ and $y$ are real Gaussian processes with
    covariance functions $k_r$ and $k_i$.

    First, suppose that both conditions hold.
    Since $r$ is the reproducing kernel of $H_2$ by
    Proposition~\ref{thm:h2_rkhs},
    condition~\eqref{eq:hinf-conditions-driscoll} ensures by
    Lemma~\ref{lem:driscoll}
    zero-one theorem for Gaussian processes that the sample paths of $x$ and
    $y$ lie in $H_2$ with probability one, ensuring the same for $f$.
    Since $k_r$ and $k_i$ satisfy the hypotheses of
    Lemma~\ref{lem:boundedness_condition}, it follows that $x$ and $y$ are
    bounded with probability one, which implies the same for $f$.
\end{proof}

\end{extendedonly}

\begin{remark}
    According to Driscoll's theorem, the probability that $f\in H_2$ is either
    zero or one. (Zero occurs when either supremum in
    condition~\eqref{eq:hinf-conditions-driscoll} is infinite.) Similarly, the
    realizations of a Gaussian process are bounded with probability zero or
    one~(\cite{landau1970supremum}). This means that the realizations of
    a z-domain Gaussian process are either almost surely $H_\infty$ functions or
    almost surely not.
\end{remark} 

\begin{remark}
Condition~\eqref{eq:hinf-conditions-driscoll} is necessary and sufficient for
$f$ to inhabit $H_2$ with probability one. On the other hand,
condition~\eqref{eq:hinf-conditions-boundedvar} is sufficient but not necessary
for $\normi{f}$ to be bounded. Indeed, necessary and sufficient conditions for a
stochastic process to be almost surely bounded are generally not available even
for real-valued Gaussian processes except in special cases. 
Fortunately, covariance functions in practice often satisfy a stronger condition
that implies~\eqref{eq:hinf-conditions-boundedvar}~(\cite[eq. 2.5.17]{RFGA}),
namely that
$
    k(s,t) = k(s,s) - q(s-t) + O(|s-t|^{2+\delta}),
$
for small $|s-t|$, where $q$ is a positive definite quadratic form and $\delta>0$. 
\end{remark}

The general condition for a process to be conjugate symmetric is given by the
following result.

\begin{proposition}
    \label{prop:real-impulse}
    Let $f$ be a z-domain Gaussian process with domain $D$,
    covariance $k$, and complementary covariance $\tilde{k}$. 
    Then $f$ is conjugate-symmetric
    if and only if $k$ and $\tilde{k}$ satisfy the conditions
    \begin{equation}
        \label{eq:real_impulse_conditions}
        \begin{aligned}
            k(z,z) &= k(z^*, z^*), \qquad k(z,z) = \tilde{k}(z, z^*)
        \end{aligned}
    \end{equation}
    for all $z\in D$.
\end{proposition}

\begin{shortonly}
    \begin{proofsketch}
        Under~\eqref{eq:real_impulse_conditions}, the joint distribution for $(f^*(z),f(z^*))$
        is a degenerate complex Gaussian distribution where both components are
        perfectly correlated with the same variance, and thus equal with probability one.
        If~\eqref{eq:real_impulse_conditions} doesn't hold, this cannot be true for all $z$.
    \end{proofsketch}
\end{shortonly}

\begin{extendedonly}

\begin{proof}
    For a fixed $z\in D$, consider the random vector 
    $(f^*(z),f(z^*))$. This is a multivariate complex normal whose augmented covariance matrix is
    \begin{equation}
        \begin{aligned}
        \Ex{
            \begin{bmatrix}
                f^*(z) \\ f(z^*) \\ f(z) \\ f^*(z^*)
            \end{bmatrix}
            \begin{bmatrix}
                f^*(z) & f(z^*) & f(z) & f^*(z^*)
            \end{bmatrix}
        }
        =
        \begin{bmatrix}
            k^*(z,z) & \tilde{k}^*(z,z^*) & \tilde{k}^*(z,z) & k^*(z,z^*) \\
            \tilde{k}(z^*,z) & k(z^*,z^*) & k(z^*,z) & \tilde{k}(z^*,z^*) \\
            \tilde{k}(z,z) & k(z,z^*) & k(z,z) & \tilde{k}(z,z^*) \\
            k^*(z^*,z) & \tilde{k}^*(z^*,z^*) & \tilde{k}^*(z^*,z) & k^*(z^*,z^*)
        \end{bmatrix}.
        \end{aligned}
    \end{equation}
    Since augmented covariance matrices have the form
    \begin{equation}
        \begin{bmatrix}
            K & \tilde{K} \\ \tilde{K}^* & K^*
        \end{bmatrix}
    \end{equation}
    where $K$ and $\tilde{K}$ are the Hermitian and complementary covariances,
    the Hermitian and complementary covariances of $(f^*(z),f(z^*))$ are
    \begin{equation}
        K
        =
        \begin{bmatrix}
            k^*(z,z) & \tilde{k}^*(z,z^*) \\
            \tilde{k}(z^*,z) & k(z^*,z^*)
        \end{bmatrix},
        \qquad
        \tilde{K}
        =
        \begin{bmatrix}
            \tilde{k}^*(z,z) & k^*(z,z^*) \\
            k(z^*,z) & \tilde{k}(z^*,z^*)
        \end{bmatrix}.
    \end{equation}
    Since $(f^*(z),f(z^*))$ is Gaussian with mean zero, this means that
    \begin{equation}
        \label{eq:f_fstar_distribution}
        \begin{bmatrix}
            f^*(z) \\ f(z^*)
        \end{bmatrix}
        \sim
        \CNorm{
            \begin{bmatrix}
                0 \\ 0
            \end{bmatrix},
            \begin{bmatrix}
                k^*(z,z) & \tilde{k}^*(z,z^*) \\
                \tilde{k}(z^*,z) & k(z^*,z^*)
            \end{bmatrix},
            \begin{bmatrix}
                \tilde{k}^*(z,z) & k^*(z,z^*) \\
                k(z^*,z) & \tilde{k}(z^*,z^*)
            \end{bmatrix}
        }.
    \end{equation}
    Under the conditions given on $k$ and $\tilde{k}$, this reduces to
    \begin{equation}
        \label{eq:f_fstar_reduction}
        \begin{bmatrix}
            f^*(z) \\ f(z^*)
        \end{bmatrix}
        \sim
        \CNorm{
            \begin{bmatrix}
                0 \\ 0
            \end{bmatrix},
            k(z,z)
            \begin{bmatrix}
                1 & 1 \\ 1 & 1
            \end{bmatrix},
            \tilde{k}(z,z)
            \begin{bmatrix}
                1 & 1 \\ 1 & 1
            \end{bmatrix}
        },
    \end{equation}
    from which it follows that
    \begin{equation}
        f^*(z) - f(z^*) =
        \begin{bmatrix}
            1 & -1
        \end{bmatrix}
        \begin{bmatrix}
            f^*(z) \\ f(z^*)
        \end{bmatrix}
        \sim
        \CNorm{0, 0, 0}.
    \end{equation}
    This means that $f^*(z) - f(z^*)=0$, or equivalently $f^*(z) = f(z^*)$, with
    probability one, for all $z\in D$.
    On the other hand, if the conditions
    in~\eqref{eq:real_impulse_conditions} are not met, then for at least one
    $z\in D$, the reduction from~\eqref{eq:f_fstar_distribution}
    to~\eqref{eq:f_fstar_reduction} is not possible, in which case
    $f(z^*)=f^*(z)$ does not hold.
\end{proof}

\end{extendedonly}

Together, Theorem~\ref{prop:hinf-conditions} and Proposition~\ref{prop:real-impulse}
give sufficient conditions on the covariance functions of a
general mean-zero z-domain Gaussian process in order for it to be a
conjugate-symmetric $H_\infty$ Gaussian
process. While Conditions~\eqref{eq:hinf-conditions-boundedvar}
and~\eqref{eq:real_impulse_conditions} can be verified in practice,
Condition~\eqref{eq:hinf-conditions-driscoll} generally cannot. We are
therefore motivated to find special cases of z-domain Gaussian processes for
which~\eqref{eq:hinf-conditions-driscoll} can be replaced by a more tractable
condition.
The broadest such case that we have found is where, in addition to satisfying
Conditions~\eqref{eq:hinf-conditions-boundedvar}
and~\eqref{eq:real_impulse_conditions}, the Hermitian covariance function is
stationary when restricted to the unit circle.
\begin{definition}
    A z-domain Gaussian process is \emph{Hermitian stationary} when its
    Hermitian covariance function satisfies
    $k(e^{j\theta}, e^{j\phi})=k(e^{j(\theta-\phi)},1)$ for all $\theta,\phi\in[-\pi,\pi)$. 
\end{definition}
Using a stationary process as a prior is common practice in
machine learning and control-theoretic applications of Gaussian process models.
Stationary processes are useful for constructing regression priors that do not
introduce unintended biases in their belief about the frequency response: since
$f(e^{j\theta})$ has the same
Hermitian variance across the entire unit circle, a sample path from a Hermitian
stationary $H_\infty$ process is just as likely to exhibit low-pass behavior as
it is high-pass or band-pass.%
\footnote{To be truly ``noninformative'' in the sense of introducing unwanted
biases, the complementary covariance should be stationary.
However, this is not possible while
satisfying~\eqref{eq:real_impulse_conditions}.}
We can obtain a ``partially informative'' prior by adding an $H_\infty$ process encoding
strong beliefs in one frequency range (such as the presence of a resonance peak)
to an $H_\infty$ process encoding weaker beliefs across all frequencies. The
sum, also an $H_\infty$ process, encodes a combination of these beliefs.

Under the additional condition of Hermitian stationarity, it turns out that the
$H_\infty$ process is characterized by a sequence of nonnegative constants.

\begin{theorem}
    \label{prop:stationary-process}
    Let $f$ be a Hermitian stationary, conjugate-symmetric z-domain Gaussian
    process with continuous Hermitian covariance $k$ and complementary
    covariance $\tilde{k}$.
    Then $f$ is an $H_\infty$ process if and only if $k$ and $\tilde{k}$ have
    the form
    \begin{equation}
        \label{eq:covariance_expansion}
        k(z,w) = \sum_{n=0}^\infty a_n^2 (zw^*)^{-n},
        \qquad
        \tilde{k}(z,w) = \sum_{n=0}^\infty a_n^2 (zw)^{-n},
    \end{equation}
    where $\{a_n\}_{n=0}^\infty$ is a nonnegative real $\ell^1$ sequence.
    Furthermore, $f$ may be expanded as \begin{equation}
        \label{eq:process_expansion}
        f(z) = \sum_{n=0}^\infty a_n w_n z^{-n},
    \end{equation}
    where $w_n\simiid\Norm{0,1}$.
\end{theorem}

\begin{shortonly}
    \begin{proofsketch}
        That~\eqref{eq:process_expansion} leads
        to~\eqref{eq:covariance_expansion} follows from direct calculation and
        the independence of the $w_n$. Under the summability condition on the
        $a_n$, the impulse response of $f$ is absolutely summable with
        probability one, implying BIBO stability.
        Thus~\eqref{eq:process_expansion} is BIBO stable (and hence in
        $H_\infty)$ with probability one.

        To see that a Hermitian stationary, conjugate-symmetric $H_\infty$
        Gaussian process must
        have~\eqref{eq:covariance_expansion},~\eqref{eq:process_expansion} and
        that the summability condition holds,
        start with an expansion of $f$ into the basis $\{z^{-n}\}_{n=0}^\infty$ of $H_\infty$.
        Hermitian stationarity implies that the coefficients are uncorrelated,
        turning the basis function expansion into~\eqref{eq:process_expansion},
        from which~\eqref{eq:covariance_expansion} follows. Since the process is
        $H_\infty$, its impulse response must be absolutely summable with
        probability one, which is only true if the summability condition on the
        $a_n$ holds.
    \end{proofsketch}
\end{shortonly}

\begin{extendedonly}

\begin{proof}
    First, we show that $f$ having the form~\eqref{eq:process_expansion} with
    positive
    $\{a_n\}_{n=0}^\infty\in\ell^1$ implies
    that $f$ is a Hermitian stationary $H_\infty$ process
    satisfying~\eqref{eq:real_impulse_conditions} with the given covariances.
    Suppose that $f(z) = \sum_{n=0}^\infty a_n w_n z^{-n}$.
    We can readily see that
    \begin{equation}
        \begin{aligned}
            k(z,w) 
            &= \Ex{f(z)f^*(w)}
            = \Ex{
                \left(
                    \sum_{n=0}^\infty a_n w_n z^{-n}
                \right)
                \left(
                    \sum_{m=0}^\infty a_m w_m (w)^{-m}
                \right)^*
            }\\
            &= \Ex{
                \sum_{n=0,m=0}^\infty
                a_n a_m w_n w_m z^{-n} (w^*)^{-n}
            }
            = \sum_{n=0}^\infty a_n^2 (zw^*)^{-n},
        \end{aligned}
    \end{equation}
    where the cross terms vanish by the independence of the $w_n$. 
    Note also that
    $k(z,z)=k(z^*,z^*)$, which is the first part of
    condition~\eqref{eq:real_impulse_conditions}.
    A similar calculation yields
    \begin{equation}
        \tilde{k}(z,w) = \Ex{f(z)f(w)} 
        = \sum_{n=0}^\infty a_n^2 (zw)^{-n}=k(z,w^*),
    \end{equation}
    showing that both parts of condition~\eqref{eq:real_impulse_conditions} are
    satisfied,
    and that $f$ has real impulse response $ h_f(n) = w_n a_n$.
    Recall that a SISO system is BIBO stable if its impulse response
    is absolutely summable. 
    To that end, consider the sequence
    \begin{equation}
        M_T=
        \sum_{n=0}^T |h_f(n)|
        =\sum_{n=0}^T a_n |w_n|
    \end{equation}
    of partial sums:
    if $\lim_{T\to\infty}M_T$ converges to a random variable that is finite with probability one, then the impulse
    response is absolutely summable with probability one. Since the $w_n$ are independent and 
    $0 < a_n |w_n| < \infty $ for all $n$, it follows that $M_T$ is a submartingale and that
    $\Ex{M_T}$ increases monotonically. Using the summability condition on the
    $a_n$ and the fact that $\Ex{|w_n|}=\sqrt{2/\pi}$ (as $|w_n|$ follows a
    half-normal distribution), we have
    \begin{equation}
        \Ex{\sum_{n=0}^\infty |h_f(n)|} = \sum_{n=0}^\infty a_n |w_n|
        =\sqrt{\frac{2}{\pi}}\sum_{n=0}^\infty a_n<\infty,
    \end{equation}
    which means $\sup_T \Ex{M_T}<\infty$ by monotonicity. Since $M_T$ is a
    submartingale and $\sup_T \Ex{M_T}$ is finite, it follows by the
    Martingale convergence theorem~\cite[Theorem 4.2.11]{durrett2019probability}
    that the limit of $M_T$ converges to a
    random variable that is finite with probability one.
    This shows that $h_f$ is absolutely summable with probability one, implying
    BIBO stability and that $f\in H_\infty$ with probability one.

    Next, we show that a Hermitian stationary, conjugate-symmetric $H_\infty$
    Gaussian process $f$ must have Hermitian and complementary covariances
    of the form~\eqref{eq:covariance_expansion}, and that this in turn implies
    that $f$ has the form~\eqref{eq:process_expansion}. Since $f\in H_\infty$,
    we can use the fact that $z^{-n}, n \ge 0$ is a basis for $H_\infty$ to
    expand $f$ as
    \begin{equation}
        f(z) = \sum_{n=0}^\infty h_n z^{-n},
    \end{equation}
    where the coefficients $h_n = \langle f,z^{-n}\rangle_2$ are an infinite
    sequence of random variables.
    Since $f$ is Gaussian and conjugate symmetric,
    the $h_n$ are real Gaussian random
    variables that may be correlated. From this form, we can express the
    Hermitian and complementary covariance as
    \begin{equation}
        \begin{aligned}
            k(z,w) &= \sum_{n=0}^\infty \sum_{m=0}^\infty 
                \Ex{h_n h_m} z^{-n}(w^*)^{-m}\\
            \tilde{k}(z,w) &= \sum_{n=0}^\infty \sum_{m=0}^\infty 
                \Ex{h_n h_m} z^{-n}(w)^{-m},
        \end{aligned}
    \end{equation}
    which shows that $\tilde{k}(z,w)=k(z,w^*)$ for $z,w\in E^2$. Restricting the
    covariance functions to the unit circle, we have
    \begin{equation}
        \begin{aligned}
            \label{eq:covars1}
            k(e^{j\theta},e^{j\phi}) &= \sum_{n=0}^\infty \sum_{m=0}^\infty 
                \Ex{h_n h_m} e^{-j(n\theta-m\phi)}\\
            \tilde{k}(e^{j\theta},e^{j\phi}) &= \sum_{n=0}^\infty \sum_{m=0}^\infty 
                \Ex{h_n h_m} e^{-j(n\theta+m\phi)}.\\
        \end{aligned}
    \end{equation}
    By the assumption of Hermitian stationarity, we know that
    $k(e^{j(\theta-\phi)}, 1)$ is a positive definite function whose domain is
    the unit circle. We can therefore apply Bochner's 
    theorem~\cite[section 1.4.3]{rudin1962fourier} to obtain a second expansion
    \begin{equation}
            \label{eq:covars2}
        k(e^{j(\theta-\phi)}, 1)=
        k(e^{j\theta}, e^{j\phi})=
        \sum_{n\in\Z} a_n^2 e^{-jn(\theta - \phi)},
    \end{equation}
    where $a_n$ are real and nonnegative. In order for the
    expansion of $k$ in~\eqref{eq:covars1} and the expansion
    in~\eqref{eq:covars2} to be equal, the positive-power terms
    in~\eqref{eq:covars2} must vanish, and the cross-terms in~\eqref{eq:covars1}
    must vanish.  

    This means that $\Ex{h_n h_m}=0$ for $m\ne n$,
    from which it follows that the covariances have the form
    \begin{equation}
        \begin{aligned}
            k(z,w) &= 
            \sum_{n=0}^\infty
            \Ex{h_n^2} (zw^*)^{-n}
            =
            \sum_{n=0}^\infty
            a_n^2 (zw^*)^{-n}
            \\
            \tilde{k}(z,w) &= 
            \sum_{n=0}^\infty
            \Ex{h_n^2} (zw)^{-n}
            =
            \sum_{n=0}^\infty
            a_n^2 (zw)^{-n}
        \end{aligned}
    \end{equation}
    where we identify $a_n^2=\Ex{h_n^2}$,
    and that the $h_n$ are independent. Returning to the expanded form of the
    process and expressing $\Ex{h_n^2}=a_n^2$, we have
    \begin{equation}
        f(z) = \sum_{n=0}^\infty w_n a_n z^{-n}
    \end{equation}
    where $w_n\simiid\Norm{0,1}$.

    Evidently, the impulse response $h_f$ has the same form as
    before, so the expected absolute sum of the impulse response is
    $\Ex{\sum_{n=0}^\infty |h_f(n)|}=\sqrt{\frac{2}{\pi}}\sum_{n=0}^\infty a_n$.
    Since $f\in H_\infty$ by assumption, it follows that 
    $\sum_{n=0}^\infty |h_f(n)|$ almost surely converges, and therefore that
    $\Ex{\sum_{n=0}^\infty |h_f(n)|} <\infty$ by the Kolmogorov three-series
    theorem~(\cite[Theorem 2.5.8]{durrett2019probability}, condition (ii)), 
    showing that $\{a_n\}_{n=0}^\infty\in\ell^1$.

\end{proof}

\end{extendedonly}

Theorem~\ref{prop:stationary-process} provides a useful tool for
constructing conjugate-symmetric $H_\infty$ Gaussian processes: all we need to
do is select a summable sequence of nonnegative numbers. 
\begin{shortonly}
For example, we use
Theorem~\ref{prop:stationary-process} to construct the following regression prior for the next section.
\end{shortonly}

\begin{example}[Geometric $H_\infty$ process]
    \label{ex:geometric}
    Take $a_n^2 = \alpha^n$ with $\alpha\in(0,1)$; this yields
    a conjugate-symmetric $H_\infty$ Gaussian process with Hermitian covariance
    $k_\alpha(z,w)=\sum_{n=0}^\infty \alpha^n (zw^*)^{-n} = \frac{zw^*}{zw^*-\alpha}$
    and complementary covariance
    $\tilde{k}_\alpha(z,w) = \frac{zw}{zw-\alpha}$.
\end{example}

\begin{extendedonly}
\begin{example}[Exponential $H_\infty$ process]
    \label{ex:exponential}
    Take $a_n^2 = \frac{1}{n!}$; 
    this yields a conjugate-symmetric
    $H_\infty$ Gaussian process with Hermitian covariance
    $k(z,w) = \sum_{n=0}^\infty \frac{(zw^*)^{-n}}{n!} = e^{-zw^*}$
    and complementary covariance
    $\tilde{k}(z,w) = e^{-zw}$.
\end{example}
\end{extendedonly}

\section{Gaussian Process Regression in the Frequency Domain}
\label{sec:regression}

Let $H_g\in H^\infty$ denote the system whose transfer function $g\in H_\infty$
we wish to identify. While not stochastic, $g$ is unknown, and we represent both
our uncertainty and our prior beliefs in a Bayesian fashion with an $H_\infty$ Gaussian process
$f$ with
Hermitian and complementary covariances $k$ and $\tilde{k}$. To model our prior
beliefs, the distribution of $f$ should give greater probability to functions we
believe are likely to be similar to $g$, and should assign probability zero to
functions we know that $g$ cannot be. As an example of the latter, the fact that
$P(f\in H_\infty)=1$ encodes our belief that $g\in H_\infty$, which demonstrates
the importance of $H_\infty$ Gaussian processes for prior model design. 

We suppose that our data consists of $n$ noisy frequency-domain point estimates
$y_i = g(z_i) + e_i$, where $e_i\simiid\Norm{0,\sigma_e^2}$, $z_i\in\bar{E}$.
If our primary form of data is a time-domain trace of input and output values,
we first convert this data into an \emph{empirical transfer function estimate}
(ETFE). There are several well-established methods to construct ETFEs from time
traces, such as Blackman-Tukey spectral analysis, windowed filter banks, or
simply dividing the DFT of the output trace by the DFT of the input trace. In
our numerical examples, we will use windowed filter banks.

Our approach is essentially the same procedure as standard Gaussian process
regression as described in~\cite{GPML} extended to the complex case.
We take the mean of the prior model to be zero without loss of generality. 
To estimate the transfer function at a new point $z$, we note that $g(z)$ is
related to $(y_1,\dotsc,y_n)$ under the prior model as
\begin{equation}
    \begin{bmatrix}
        g(z)
        \\ 
        y 
    \end{bmatrix}
    \sim
    \CNorm{0,
        \begin{bmatrix}
            K_{xx} & K_{xy} \\
            K_{xy}^H & K_{yy} \\
        \end{bmatrix},
        \begin{bmatrix}
            \tilde{K}_{xx} & \tilde{K}_{xy} \\
            \tilde{K}_{xy}^H & \tilde{K}_{yy} \\
        \end{bmatrix}
    };
\end{equation}
where $y\in\C^n$, $K_{yy}\in\C^{n\times n}$,
$K_{xy}\in\C^{n\times 1}$, and
$K_{xx}\in\C$ are defined componentwise as
\begin{equation}
    (y)_i=y_i,
    \quad
    \left(K_{yy}\right)_{ij} = k(z_i,z_j) + \sigma_e^2\delta_{ij},
    \quad
    \left(K_{xy}\right)_{ij} = k(z,z_i),
    \quad
    K_{xx} = k(z,z) + \sigma_e^2,
\end{equation}
and the components of the complementary covariance matrix are defined
analogously. 
\begin{shortonly}
The minimum-error linear estimator
$\hat{g}(z)$ for $g(z)$ given the data and its predictive Hermitian
variance $\sigma^2_g$ are~(\cite[\S 5.3]{schreier2010statistical})
\begin{equation}
    \label{eq:strictly_linear_prediction}
    \hat{g}(z) = K_{xy}^H K_{yy}^{-1} y,
    \qquad
    \sigma^2_g(z) = k(z,z) - K_{xy}^H K_{yy}^{-1}k_{xy}.
\end{equation}
These expressions are identical to the posterior mean and variance of a real
Gaussian process regression model (cf. Equation (2.19) in~\cite{GPML}) except
that $K_{xx}$, $K_{xy}$, and $K_{yy}$ are complex-valued.

For general complex-valued Gaussian regression, the \emph{widely linear}
estimator, which incorporates $y^*$ and the complementary covariance, is an
improvement over the strictly linear estimator. The 
degree of improvement is measured by the matrix 
$P=K_{yy} - \tilde{K}_{yy}(K_{yy}^*)^{-1}\tilde{K}^*_{yy}$, 
which is the error variance of
linearly estimating $y^*$ from $y$ under the prior model.
In particular, when $P=0$ the strictly and widely linear estimators
coincide~\cite[\S 5.4.1]{schreier2010statistical}. In our experiments, we find that the strictly linear estimator
performs well for conjugate-symmetric $H_\infty$ GP priors, and that $P$ is
nearly singular and very small in norm compared to $K_{yy}$ and
$\tilde{K}_{yy}$, implying that its performance is close to the widely linear
estimator.%
\footnote{The widely linear regression equations, which we show in~(EXTENDED
PAPER), require inverting $P$. In this case, the
widely linear estimator is numerically unstable compared to the strictly linear
estimator.}
For these reasons, we use the strictly linear estimator in the
regression examples below.
We believe the strictly linear estimator works well for conjugate-symmetric
$H_\infty$ process because of the prior assumption of causality and conjugate
symmetry. We discuss this in more detail in~(EXTENDED PAPER).

\end{shortonly}

\begin{extendedonly}

By conditioning $f(z)$ on the data $y$ according to the prior
model, we obtain the posterior distribution of $f(z)$. According to the
conditioning law for multivariate complex Gaussian random 
variables~\cite[\S 2.3.2]{schreier2010statistical}, this is
$f(z)|y\sim\CNorm{\mu, \sigma_p^2, \tilde{\sigma}_p^2}$, where
\begin{equation}
    \begin{aligned}
        \label{eq:widely_linear_prediction}
        \mu_p
        &=
        (K_{xy} - \tilde{K}_{xy}(K^*_{yy})^{-1}\tilde{K}_{yy}^*)P^{-1}y
        +
        (\tilde{K}_{xy}-K_{xy}K_{yy}^{-1}\tilde{K}_{yy})(P^*)^{-1}y^*\\
        \sigma_p^2
        &= k_{zz}
        - K_{xy}P^{-1}K_{xy}^H 
        + \tilde{K}_{xy}K_{yy}^{-1}\tilde{K}_{yy}(P^*)^{-1}K_{xy}^H\\
        &\quad
        - \tilde{K}_{xy}(P^*)^{-1}\tilde{K}_{xy}^H
        + K_{xy}K_{yy}^{-1}\tilde{K}_{yy}(P^*)^{-1}\tilde{K}_{xy}^H\\
        \tilde{\sigma}_p^2
        &= \tilde{k}_{zz}
        - K_{xy}P^{-1}(\tilde{K}_{xy}^*)^H
        + \tilde{K}_{xy}K_{yy}^{-1}\tilde{K}_{yy}(P^*)^{-1}(\tilde{K}_{xy}^*)^H\\
        &\quad
        - \tilde{K}_{xy}(P^*)^{-1}(K_{xy}^*)^H
        + K_{xy}K_{yy}^{-1}\tilde{K}_{yy}(P^*)^{-1}(K_{xy}^*)^H
    \end{aligned}
\end{equation}
and where $P$ denotes the Schur product 
$P=K_{yy} - \tilde{K}_{yy}(K_{yy}^*)^{-1}\tilde{K}^*_{yy}$.
The predictive mean $\mu_p$ is the \emph{minimum mean-square error widely linear
estimator} of $g(z)$ given $y$, where ``widely linear'' means
that $\mu_p$ is a linear combination of both $y$ and $y^*$. 
A \emph{strictly linear} estimator, on the other hand, uses only $y$. Under the
same circumstances as above, the minimum least-square strictly linear estimator
for $g(z)$ given $y$ and its error variance are respectively
\begin{equation}
    \label{eq:strictly_linear_prediction}
    \hat{g}(z) = K_{xy}^H K_{yy}^{-1} y,
    \qquad
    \sigma^2_g(z) = K_{zz} - K_{xy}^H K_{yy}^{-1}k_{xy},
\end{equation}
which are identical to the posterior mean and variance of a real Gaussian
process regression model (cf. Equation (2.19) in~\cite{GPML}) except that $K_{xx}$, $K_{xy}$, and
$K_{yy}$ are complex-valued.

The widely linear estimator can only be an improvement on the linear estimator,
since an estimate made using $y$ can certainly be made using
$(y,y^*)$. The improvement is measured by the Schur
complement $P$ defined above, which is the error covariance of statistically
estimating $y^*$ from $y$, or equivalently estimating the real part given the
imaginary part.
In particular, when $P=0$, the strictly linear and widely linear
estimators coincide, and the expressions
in~\eqref{eq:widely_linear_prediction} become ill-defined.
One case where this holds is when the covariances are \emph{maximally improper},
in which case the imaginary part can be estimated from the real with zero error.

In our experiments with real-impulse $H_\infty$ processes, we have found that
$P$ tends to be close to singular, and small in induced 2-norm and Frobenius norm
relative to $K_{yy}$ and $\tilde{K}_{yy}$. This makes the mean and variance
computations in~\eqref{eq:widely_linear_prediction} numerically unstable while
also implying that the strictly linear estimator will perform similarly to the
widely linear estimator. 
We believe this is due to the symmetry
condition imposed on $k$ and $\tilde{k}$ by having real impulse response. 
This condition implies that the imaginary part can be computed exactly from the
real part by the discrete Hilbert transform~\cite[\S 2.26]{rabiner1975theory}.
The covariance matrices $K_{yy}$ and $\tilde{K}_{yy}$ will not themselves be
maximally improper, since the Hilbert transform requires knowledge over the
entire unit circle; however, our experiments suggest that they are close to
maximally improper, and we conjecture that they become maximally improper in the
limit of infinite data.
This suggests that the strictly linear estimator will perform well for
conjugate-symmetric $H_\infty$ priors. For this reason, as well as the numerical
instability of the widely linear estimator when $P$ is close to singular, we use
the strictly linear estimator in our numerical experiments.

\end{extendedonly}

For $z\in D$ and $\eta > 0$, define the \emph{confidence ellipsoid} 
$\mathcal{E}_\eta(z) = \{w\in \C: |w-\hat{g}(z)|^2 \le \eta^2\sigma^{2}_g(z)\}$. By
Markov's inequality, we know that $f(z)\in\mathcal{E}_\eta(z)$ with probability
$\ge 1 - 1/\eta^2$. 
This implies bounds on the real and imaginary parts
by projecting the confidence ellipsoid onto the real and imaginary axes:
from these we can construct probabilistic bounds on the magnitude
and phase of $f(z)$ via interval arithmetic, which we will see in the numerical examples.

Let $\theta\in\Theta$ denote the hyperparameters of a covariance
function $k_\theta$, so that $K_{yy}$ becomes a function of $\theta$:
then the log marginal likelihood of the data under the posterior for the strictly
linear estimator is
$
    L(\theta) = -\tfrac{1}{2}\left(
        y^HK_{yy}(\theta)^{-1}y + \log\det K_{yy}(\theta) + n\log 2\pi
    \right).
$
Keeping the data $y$ and input locations $z_i$ fixed, $L(\theta)$ measures the
probability of observing data $y$ when the prior covariance function is
$k_\theta$. By maximizing $L$ with respect to $\theta$, we find the covariance
among $k_\theta$, $\theta\in\Theta$ that best explains the observations.%
\footnote{Although it seems contradictory to choose prior parameters based on
posterior data, it can be justified as an empirical-Bayes approximation to a
hierarchical model with $\theta$ as hyperparameter. }

\begin{extendedonly}
To summarize, the regression process is as follows:
\begin{enumerate}
    \item Select a family of $H_\infty$ Gaussian process models indexed by a
        hyperparameter set $\Theta$;
    \item observe point estimate data of $g(z)$, typically by an empirical transfer function estimate;
    \item Select $\theta\in\Theta$ that maximizes the log likelihood $L(\theta)$;
    \item Use the strictly linear
        estimator~\eqref{eq:strictly_linear_prediction}
        to obtain an estimate $\hat{g}$ and predictive
        variance $\sigma_g$.
\end{enumerate}
We now demonstrate the process by identifying two second-order systems.
\end{extendedonly}

\subsection{Examples: Identifying Second-order Systems}%
\label{sub:examples}

\begin{figure}[htpb]
    \centering
    \includegraphics[width=\linewidth]{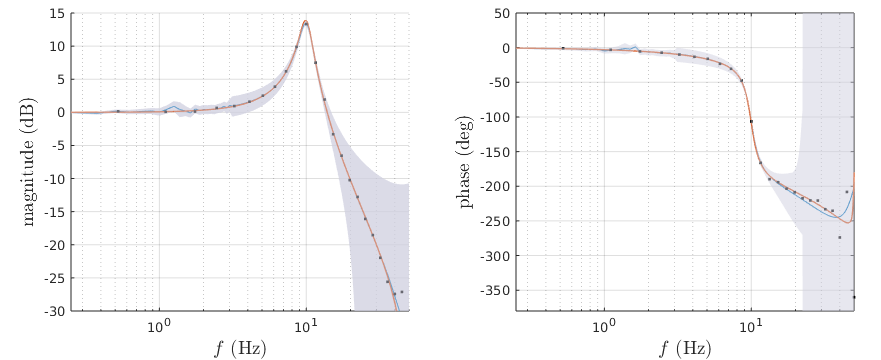}
    \caption{Bode plot of the second-order resonant system (orange), and its estimate
    (blue) using $H_\infty$ Gaussian process regression from an empirical
transfer function estimate (black points) with $\eta=3$ confidence ellipsoid bounds (grey).}%
    \label{fig:resonance_regression}
\end{figure}

We apply the strictly-linear $H_\infty$ Gaussian process regression method
described above to the problem of identifying two second-order systems.
The first test system is a second-order system that exhibits a resonance peak.
The system is specified in continuous time, with canonical second-order transfer
function
$    g(s) = \frac{\omega_0^2}{s^2+2\xi\omega_0s + \omega_0^2}$,
where $\omega_0=20\pi$ rad/s, and $\xi=0.1$, and converted to the discrete-time
transfer function $g(z)$ using a zero-order hold discretization with a sampling
frequency of $f_s=100$ Hz. 
We suppose that we know \emph{a priori} that there is a resonance peak, but not
about its location or half-width, and we have no other strong information about
the frequency response. For this prior belief, an appropriate prior model is a
weighted mixture of a cozine process and a Hermitian stationary process.
In particular, we use the family of $H_\infty$ processes with covariance
functions
$
    k(z,w) = \sigma^2_g k_g(z,w) + \sigma^2_c k_c(z,w),
    \tilde{k}(z,w) = \sigma^2_g \tilde{k}_g(z,w) + \sigma^2_c \tilde{k}_c(z,w),
$
where $k_g$ is the covariance of the geometric $H_\infty$ process defined in
Example~\ref{ex:geometric}, and $k_c$ is the covariance of the cozine process,
and likewise for the complementary covariance. $\sigma^2_g$ and $\sigma^2_c$ are
weights that determine the relative importance of the two
parts of the model. This family of covariances has five hyperparameters:
$k_g\in[0,\infty)$,
$\alpha\in(0,1)$,
$k_c\in[0,\infty)$,
$\omega_0\in[0,\pi]$, and
$a\in(0,1)$.

We suppose that an input trace $u(n)$ of Gaussian white noise with variance $\sigma^2_u=1/f_s$ is run through
$H_g$ yielding an output trace $y(n)$; our observations comprise these two
traces, corrupted by additive Gaussian white noise of variance
$\sigma^2=10^{-4}/f_s$. To obtain an empirical transfer function estimate,
we run both observation traces through a bank of 25 windowed 1000-tap DFT filters. 
The impulse responses of the filter bank are
$h_i(n) = e^{j\omega_i n} w(n)$ for $i=1,\dotsc,25$,
with Gaussian window $w(n)=\exp(-\tfrac{1}{2}(\sigma_w(n-500) / 1000)^2)$ for
$n=0,\dotsc,999$, and $w(n)=0$ otherwise, with window half-width
$\sigma_w=0.25$.
Let $u_i$, $y_i$ denote the outputs of filter $h_i$ with inputs $u$, $y$
respectively: $y_i(n)/u_i(n)$ gives a running estimate of $g(e^{j\omega_i})$,
whose value after 1000 time steps we take as our observation at
$z_i=e^{j\omega_i}$.
Figure~\ref{fig:resonance_regression} shows the regression from the strictly
linear estimator~\eqref{eq:strictly_linear_prediction} after tuning the
covariance hyperparameters via maximum likelihood, along with predictive error
bounds based on $\eta=3$ confidence ellipsoids.

The second is a second-order allpass filter. This system is specified in
discrete time with the transfer function
$    g(z) = \frac{|z_0|^2 - 2\Re{z_0} + 1}{1 - 2\Re{z_0} + |z_0|^2}$,
where $z_0=0.5e^{\pm j\pi/4}$ are the system's poles, with sampling
frequency $f_s= 100$ Hz.
For this system we assume
that we do not have \emph{a priori} information on the structure of the
frequency response, so we use a Hermitian stationary $H_\infty$ process as the
prior model. In particular, we take the family of geometric $H_\infty$ process,
indexed by hyperparameter $\alpha\in(0,1)$. To construct the
empirical transfer function estimate, we use the same data model and filter bank
as the previous example. 
Figure~\ref{fig:allpass_regression} shows the strictly linear regression after tuning the
covariance hyperparameters, again with predictive error bounds from
$\eta=3$ confidence ellipsoids.

\begin{figure}[htpb]
    \centering
    \includegraphics[width=\linewidth]{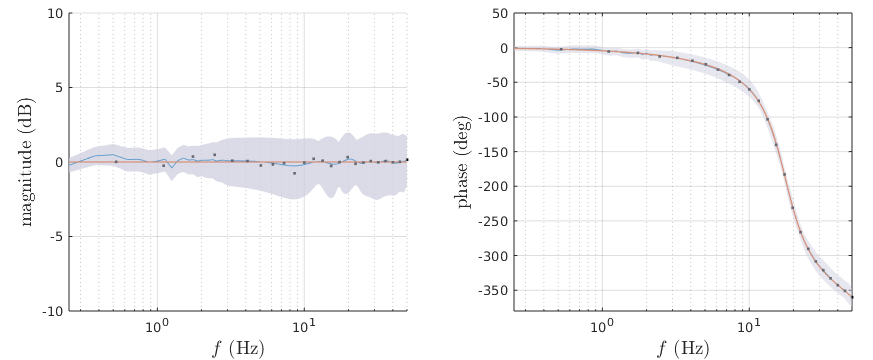}
    \caption{Bode plot of the second-order allpass system (orange), and its estimate
    (blue) using $H_\infty$ Gaussian process regression from an empirical
transfer function estimate (black points) with $\eta=3$ confidence ellipsoid bounds (grey).}%
\label{fig:allpass_regression}
\end{figure}

\section{Conclusion}%
\label{sec:conclusion}

The $H_\infty$ processes constructed using the results of this paper,
particularly Theorem~\ref{prop:stationary-process}, are effective priors for
Bayesian nonparametric identification of transfer functions. Furthermore, the
strictly linear estimator, which is suboptimal for general complex Gaussian
process priors, provides 
transfer function estimates that are close to optimal
for conjugate-symmetric $H_\infty$ priors. We have numerical evidence that
suggests that as the number of
frequency data points increases, the covariance becomes \emph{maximally
improper}, a case in which the strictly linear is indeed optimal. We will
investigate this conjecture in future work.

The applications presented in this paper use $H_\infty$
Gaussian process as statistically interpretable regression priors, but do not
consider questions of probabilistic robustness. We intend to follow this work
with a similar investigation into the robustness properties of $H_\infty$
models, such as probabilistic bounds on the $H_\infty$ norm, and integral
quadratic constraints that hold with high probability for an $H_\infty$ process
with given mean and covariance functions. 
\begin{extendedonly}
A simple bound on $\normi{f}$ can be
constructed using Borell's inequality on the real and imaginary parts of $f$,
but a bound tight enough for practical use will require more careful analysis.
\end{extendedonly}

\acks{%
This work was supported by the grants ONR N00014-18-1-2209, AFOSR 
FA9550-18-1-0253.
}

\bibliography{refs}

\end{document}